\documentclass[12pt,english,onecolumn, 5p]{svjour3}
\usepackage[T1]{fontenc}
\usepackage[latin9]{inputenc}
\usepackage{units}
\usepackage{mathrsfs}
\usepackage{url}
\usepackage{amsmath}
\usepackage{amssymb}

\makeatletter
\usepackage[vlined,plain,noresetcount]{algorithm2e} 
\SetAlgoSkip{smallskip}
\usepackage{multicol}

\DeclareMathOperator{\maxid}{\it max}

\usepackage{transparent}
\usepackage[adobe-utopia]{mathdesign}

\usepackage[margin=1in]{geometry}

\makeatother

\usepackage{babel}

\begin{document}

\title{Notes on discrepancy in the pairwise comparisons method }

\author{Konrad Ku\l{}akowski}

\institute{AGH University of Science and Technology, al. Mickiewicza 30, Kraków,
Poland}
\maketitle
\begin{abstract}
The pairwise comparisons method is a convenient tool used when the
relative order among different concepts (alternatives) needs to be
determined. One popular implementation of the method is based on solving
an eigenvalue problem for the pairwise comparisons matrix. In such
cases the ranking result the principal eigenvector of the pairwise
comparison matrix is adopted, whilst the eigenvalue is used to determine
the index of inconsistency. A lot of research has been devoted to
the critical analysis of the eigenvalue based approach. One of them
is the work of Bana e Costa and Vansninck \cite{BanaeCosta2008acao}.
In their work authors define the conditions of order preservation
(COP) and show that even for a sufficiently consistent pairwise comparisons
matrices, this condition can not be met. The present work defines
a more precise criteria for determining when the COP is met. To formulate
the criteria a discrepancy factor is used describing how far the input
to the ranking procedure is from the ranking result.
\end{abstract}

\section{Introduction }

The origins of pairwise comparisons (herein abbreviated as\emph{ }PC)
date back to the thirteenth century \cite{Colomer2011rlfa}. The contemporary
form of the method owes to \emph{Fechner} \cite{Fechner1966eop},
\emph{Thurstone} \cite{Thurstone27aloc} and Saaty \cite{Saaty1977asmf}.
The latter proposed the Analytic Hierarchy Process \emph{(AHP) - }a
hierarchical, eigenvalue based extension to the \emph{PC} theory,
which provides useful methods for dealing with a large number of criteria.
In its early stages the \emph{PC} method was a voting method \cite{Colomer2011rlfa}.
Later it was used in psychometrics and psychophysics \cite{Thurstone27aloc}.
Over time, it began to be used in decision theory \cite{Saaty2005taha},
economics \cite{Peterson1998evabt}, and other\\ fields. The utility
of the method has been confirmed by numerous examples \cite{Vaidya2006ahpa,Liberatore2008tahp,Subramanian2012aroa}.
Despite its long existence it is still an interesting subject for
resear\-chers. Some of its aspects still raise vigorous discussions
\cite{Dyer1990rota,Barzilai1994arrn,BanaeCosta2008acao} and prompt
researchers to enquire further into this area. Example of such exploration
are the \emph{Rough Set} approach \cite{Greco2011fkSIMPL}, fuzzy
\emph{PC} relation handling \cite{Mikhailov2003dpff,Fedrizzi2010otpv,Yuen2013fcnp},
incomplete \emph{PC} relation \cite{Bozoki2010ooco,Koczkodaj1999mnei},
non-numerical rankings \cite{Janicki2012oapc}, rankings with the
reference set of alternatives \cite{Kulakowski2013ahre,Kulakowski2013hreaCoRR}
and others. A more thorough discussion of the \emph{PC} method can
be found in \cite{Smith2004aada,Ishizaka2009ahpa}.

\section{Preliminaries }

\subsection{Pairwise comparisons method}

Central to the \emph{PC} method is a \emph{PC} matrix $M=[m_{ij}]$,
where $m_{ij}\in\mathbb{R}_{+}$ and $i,j\in\{1,\ldots,n\}$, that
expresses a quantitative relation $R$ over the finite set of concepts
$C\overset{\textit{df}}{=}\{c_{i}\in\mathscr{C}\wedge i\in\{1,\ldots,n\}\}$.
The set $\mathscr{C}$ is a non empty universe of concepts and $R(c_{i},c_{j})=m_{ij}$,
$R(c_{j},c_{i})=m_{ji}$. The values $m_{ij}$ and $m_{ji}$ are interpreted
as the relative importance, value or quality indicators of concepts
$c_{i}$ and $c_{j}$, so that according to the best knowledge of
experts $c_{i}=m_{ij}c_{j}$ \\ should hold. 
\begin{definition}
\label{def:A-matrix-recip}A matrix $M$ is said to be reciprocal
if $\forall i,j\in\{1,\ldots,n\}:m_{ij}=\frac{1}{m_{ji}}$ and $M$
is said to be consistent if $\forall i,j,k\in\{1,\ldots,n\}:m_{ij}\cdot m_{jk}\cdot m_{ki}=1$.
\end{definition}
Since the knowledge stored in the \emph{PC} matrix usually comes from
experts in the field of $R$, it may results in inaccuracy. In such
a case it may be that there exists a certain triad of values $m_{ij},m_{jk},m_{ki}$
from $M$ for which $m_{ij}\cdot m_{jk}\cdot m_{ki}\neq1$. In other
words, different ways of estimating concept value may lead to different
results. This observation gave rise to the concept of an inconsistency
index describing how far the matrix $M$ is inconsistent. There are
a number of inconsistency indexes \cite{Bozoki2008osak}. The most
popular, proposed by Saaty \cite{Saaty1977asmf}, is defined below.
\begin{definition}
\label{def:SaatyIncIdx}The eigenvalue based consistency index \emph{(Saaty's
Index)} of $n\times n$ reciprocal matrix $M$ is equal to: 
\begin{equation}
\mathcal{S}(M)=\frac{\lambda_{\textit{max}}-n}{n-1}\label{eq:Consistency_Index_AHP}
\end{equation}

where $\lambda_{\textit{max}}$ is the principal eigenvalue of $M$.
\end{definition}

The result of the pairwise comparisons method is ranking - a function
that assigns values to the concepts. Formally, it can be defined as
follows. 
\begin{definition}
The ranking function for $C$ (the ranking of $C$) is a function
$\mu:C\rightarrow\mathbb{R}_{+}$ that assigns to every concept from
$C\subset\mathscr{C}$ a positive value from $\mathbb{R}_{+}$. 
\end{definition}
In other words, $\mu(c)$ represents the ranking value for $c\in C$.
The $\mu$ function is usually written in the form of a vector of
weights $\mu\overset{\textit{df}}{=}\left[\mu(c_{1}),\ldots\mu(c_{n})\right]^{T}$.
One of the popular methods of obtaining the vector $\mu$ is to calculate
the principal eigenvector $\mu_{max}$ of $M$ (i.e. the vector associated
with the principal eigenvalue of $M$) and rescale them so that the
sum of its elements is $1$, i.e. 

\begin{align}
\mu_{\textit{ev}} & =\left[\frac{\mu_{\textit{max}}(c_{1})}{s_{\textit{ev}}},\ldots,\frac{\mu_{\textit{\textit{max}}}(c_{n})}{s_{\textit{ev}}}\right]^{T},\,\,\,\,\,\mbox{and}\nonumber \\
s_{\textit{ev}} & =\underset{i=1}{\overset{n}{\sum}}\mu_{\textit{max}}(c_{i})\label{eq:eigen-value-approach}
\end{align}

where $\mu_{\textit{ev}}$ - the ranking function, $\mu_{\textit{max}}$
- the principal eigenvector of $M$. Due to the \emph{Perron-Frobenius}
theorem \cite{Saaty1977asmf} one exists, because a real square matrix
with the positive entries has a unique largest real eigenvalue such
that the associated eigenvector has strictly positive components.

\subsection{Eigenvalue heuristics}

According to the \emph{PC} approach $m_{il}$ (an entry of $M$) should
express the relative value of $c_{i}\in C$ with respect to $c_{l}\in C$.
Therefore one would expect that $\nicefrac{\mu(c_{i})}{\mu(c_{l})}=m_{il}$,
i.e. $\mu(c_{i})=m_{il}\mu(c_{l})$ or conversely $m_{li}\mu(c_{i})=\mu(c_{l})$.
In particular, it is desirable that 
\begin{equation}
m_{li}\mu(c_{i})=\mu(c_{l})=m_{lj}\mu(c_{j})\label{eq:ideal-case}
\end{equation}
 for every $c_{i},c_{j},c_{l}\in C$. Unfortunately due to possible
data inconsistency this may not be possible, i.e. it may be the case
that $m_{li}\mu(c_{i})\neq m_{lj}\mu(c_{j})$. Therefore the question
arises of what $\mu(c_{l})$ should be? Since the values $m_{li}\mu(c_{i})$
for $i=1,\ldots,n$ can vary from each other the natural (and probably
one of the most straightforward) proposal is to adopt its arithmetic
mean as the desired candidate for $\mu(c_{l})$. This leads to the
equation: 
\begin{equation}
m_{l1}\mu(c_{1})+\ldots+m_{ln}\mu(c_{n})=n\cdot\mu(c_{l})\label{eq:first-eq}
\end{equation}

which expresses the wish that $\mu(c_{l})$ should be a compromise
between all its putative values. A natural question is whether it
is possible to achieve such a compromise for every $l=1,\ldots,n$.
In other words, whether it is possible to solve the following equation
system:
\begin{equation}
\begin{array}{ccc}
m_{11}\mu(c_{1})+\ldots+m_{1n}\mu(c_{n}) & = & n\cdot\mu(c_{1})\\
\hdotsfor[1]{3}\\
m_{n1}\mu(c_{1})+\ldots+m_{nn}\mu(c_{n}) & = & n\cdot\mu(c_{n})
\end{array}\label{eq:linear_system}
\end{equation}

This leads to the question of the solution of the following matrix
equation: 
\begin{equation}
M\mu=n\mu\label{eq:matrix-exact-eigen-value-equation}
\end{equation}

 Of course the solution of the above equation is the eigenvector
of $M$, whilst $n$ is replaced by $\lambda$ - $M$'s eigenvalue. 

\begin{equation}
M\mu=\lambda\mu\label{eq:matrix-exact-eigen-value-equation-1}
\end{equation}

There might be many eigenvectors and eigenvalues of $M$. However,
when $M$ is positive, real and reciprocal it has at least one positive
real eigenvalue associated with the positive and real eigenvector
\cite{Saaty1977asmf}. Let $\lambda_{\textit{max}}$ be the real,
largest, positive eigenvalue of $M$ and $\mu_{\textit{max}}$ be
the associated eigenvector. \emph{AHP} adopts $\mu_{\textit{max}}$
as the solution of (\ref{eq:matrix-exact-eigen-value-equation-1}).

\subsection{Local discrepancy}

In his seminal work \cite[p. 238]{Saaty1977asmf} Saaty proved the
following equality:
\begin{equation}
\lambda_{\textit{max}}-1=\sum_{i=1,i\neq j}^{n}m_{ji}\frac{\mu(c_{i})}{\mu(c_{j})}\label{eq:saaty_eq}
\end{equation}

Thus, 
\begin{equation}
\lambda_{\textit{max}}-n=\left(\sum_{i=1,i\neq j}^{n}m_{ji}\frac{\mu(c_{i})}{\mu(c_{j})}\right)-(n-1)\label{eq:saaty_eq_2}
\end{equation}

which leads to the equation describing the Saaty's inconsistency index
(Def. \ref{def:SaatyIncIdx}):

\begin{equation}
\mathcal{S}(M)=\left(\frac{1}{n-1}\sum_{i=1,i\neq j}^{n}m_{ji}\frac{\mu(c_{i})}{\mu(c_{j})}\right)-1\label{eq:saaty_eq_3}
\end{equation}

Following Saaty \cite[p. 238]{Saaty1977asmf} let us denote:
\begin{equation}
\epsilon(i,j)\overset{\textit{df}}{=}m_{ji}\frac{\mu(c_{i})}{\mu(c_{j})}=\frac{1}{m_{ij}}\frac{\mu(c_{i})}{\mu(c_{j})}\label{eq:e_def}
\end{equation}

If the ranking $\mu_{\textit{max}}$ were ideal, i.e. if each expert
judgment perfectly corresponded to the ranking results, then every
$m_{ij}$ would equal the ratio $\nicefrac{\mu(c_{i})}{\mu(c_{j})}$.
In such a case every $\epsilon(i,j)$ would equal to $1$. Otherwise,
when the ranking is imperfect the values $m_{ij}$ and $\nicefrac{\mu(c_{i})}{\mu(c_{j})}$
may vary. In other words $\epsilon(i,j)$ describes the discrepancy%
\footnote{In \cite{Saaty1977asmf} the value $\epsilon(i,l)$ is referred as
error.%
} between the particular expert judgment $m_{ij}$ and the ranking
results $\nicefrac{\mu(c_{i})}{\mu(c_{j})}$. The relationship between
$\epsilon(i,j)$ and $\mathcal{S}(M)$ (adopting $\mu_{\textit{max}}$
- the eigenvector of $M$ as the ranking function) could be written
as follows: 
\begin{equation}
\mathcal{S}(M)=\frac{1}{\left(n-1\right)}\sum_{i=1,i\neq j}^{n}\left(\epsilon(i,j)-1\right)\label{eq:SM_proof_3-2}
\end{equation}

In other words the given value of the inconsistency index $\mathcal{S}(M)$
guarantees that the arithmetic mean of the difference between assessment
accuracy determinants and one, i.e. $\epsilon(i,j)-1$ equals $\mathcal{S}(M)$.

\subsection{Conditions of Order Preservation\label{sub:COP}}

In \cite{BanaeCosta2008acao} \emph{Bana e} \emph{Costa} and \emph{Vansnick}
formulate two postulates\emph{ }(conditions of order preservation)
as regards the meaning of an eigenvalue based ranking result. The
first one, ordinal, \emph{the preservation of order preference condition}
\emph{(POP)} claims that the ranking result in relation to the given
pair of concepts $(c_{i},c_{j})$ should not break with the expert
judgement. In other words for pair of concepts $c_{1},c_{2}\in C$
such that $c_{1}$ dominates $c_{2}$ i.e. $m_{1,2}>1$ should hold
that:

\begin{equation}
\mu(c_{1})>\mu(c_{2})\label{eq:6-cop-qualitative-cond-1-1}
\end{equation}

The second one, cardinal, \emph{the preservation of order of intensity
of preference condition }(\emph{POIP),} stipulates that if $c_{1}$
dominates $c_{2}$, more than $c_{3}$ dominates $c_{4}$ (for $c_{1},\ldots,c_{4}\in C$),
i.e. if additionally \\ $m_{3,4}>1$ and $m_{1,2}>m_{3,4}$ then
also

\begin{equation}
\frac{\mu(c_{1})}{\mu(c_{2})}>\frac{\mu(c_{3})}{\mu(c_{4})}\label{eq:8-eq:cop-quantitative-cond-1}
\end{equation}

Despite the fact that the both conditions of order preservation have
been formulated in the context of eigenvalue based approach it is
important to note that, in principle, they remain valid in the context
of any priority deriving method. None of the two conditions does not
require $\mu$ be a rescaled eigenvector of $M$. Moreover, meeting
the \emph{POP} and \emph{POIP} conditions seem to be natural for any
$\mu$.

\section{The ranking discrepancy\label{sec:The-ranking-discrepancy}}

It is easy to see that $\epsilon(i,j)=\nicefrac{1}{\epsilon(j,i)}$.
For example if some $\epsilon(i,j)=2$ then $\epsilon(j,i)=0.5$.
In fact both of these values carry the same information, which is:
the ranking result for the pair $(c_{i},c_{j})$ differs twice from
the expert judgement. I.e. one concept got $100\%$ better score than
they should. The usefulness of the $\epsilon(i,j)$ parameter has
been recognized by \emph{Saaty}. For instance in \cite[p. 203]{Saaty2013dk}
the matrix $[\epsilon(i,j)]$ is used to determine which expert judgments
need to be improved in order to reduce inconsistency of $M$.

\subsection{Local discrepancy}

It turns out that $\epsilon(i,j)$ can also be used to formulate sufficient
conditions for which both \emph{COP }postulates (Def. \ref{sub:COP})
hold. For this purpose, let us define the local discrepancy $\mathcal{E}(i,j)$
as: 
\begin{equation}
\mathcal{E}(i,j)\overset{\textit{df}}{=}\max\{\epsilon(i,j)-1,\nicefrac{1}{\epsilon(i,j)}-1\}\label{eq:local_error}
\end{equation}
The value $\mathcal{E}(i,j)$ reflects local differences between ranking
results and given expert judgements. Information that for certain
$\widehat{i},\widehat{j}$ the value $\mathcal{E}(\widehat{i},\widehat{j})=0.8$
means that the discrepancy between the expert judgment $m_{\widehat{i}\widehat{j}}$
and the ranking results $\mu(c_{\widehat{i}})$ and $\mu(c_{\widehat{j}})$
reach $80\%$. Similarly as the matrix $[\epsilon(i,j)]$, also the
local discrepancy matrix $[\mathcal{E}(i,j)]$ may help to discover
where the highest discrepancy is, hence, where the expert judgement
(or the ranking function) could be improved.

\subsection{Global discrepancy}

In order to reduce (to limit) the local discrepancies it is reasonable
to introduce the concept of the global ranking discrepancy.
\begin{definition}
\label{def:local_inconsistency_def}Let the global ranking discrepancy
for the pairwise comparisons matrix $M$, and the ranking $\mu$,
be the maximal value of $\mathcal{E}(i,j)$ for $i,j=1,\ldots,n$,
i.e.

\begin{equation}
\mathcal{D}(M,\mu)\overset{\textit{df}}{=}\max_{i,j=1,\ldots,n}\mathcal{E}(i,j)\label{eq:s_loc_def}
\end{equation}

\end{definition}
Thus, a certain value of the global ranking discrepancy $\mathcal{D}(M,\mu)\leq\delta$
provides a guarantee that the maximal discrepancy between a single
assessment of an expert and the comparison of corresponding results
will not be greater than $\delta$. The ranking discrepancy $\mathcal{D}(M,\mu)$
translates directly into the inconsistency $\mathcal{S}(M)$. The
relationship can be expressed as the following theorem.
\begin{theorem}
\label{For-every-pairwise}For every pairwise comparisons matrix $M$
and the eigenvector based ranking $\mu_{\textit{max}}$ holds that:

\begin{equation}
\mathcal{D}(M,\mu_{\textit{max}})\leq\delta\Rightarrow\mathcal{S}(M)\leq\delta\label{eq:local_inc_theor}
\end{equation}
\end{theorem}
\begin{proof}
Since $\mathcal{D}(M,\mu)\leq\delta$, thus according to the definition
\ref{def:local_inconsistency_def}, every $\mathcal{E}(i,j)\leq\delta$
for $i,j=1,\ldots,n$. Thus, due to definition of $\mathcal{E}$ (see
\ref{eq:local_error}), holds that $\epsilon(i,j)-1\leq\delta$ for
every $i,j\in\{1,\ldots,n\}$. In particular for any $j\in\{1,\ldots,n\}$
it is true that:
\begin{equation}
\sum_{i=1,i\neq j}^{n}\left(\epsilon(i,j)-1\right)\leq(n-1)\delta\label{eq:local_inc_ineq}
\end{equation}

hence
\begin{equation}
\frac{1}{\left(n-1\right)}\sum_{i=1,i\neq l}^{n}\left(\epsilon(i,j)-1\right)\leq\delta\label{eq:local_inc_ineq2}
\end{equation}

which, in the light of (\ref{eq:SM_proof_3-2}) satisfies the assertion
$\mathcal{S}(M)\leq\delta$. 
\end{proof}
Hence, besides the fact that the global ranking discrepancy $\mathcal{D}(M,\mu_{\maxid})$
detects and limits the worst case discrepancy between a single expert
judgement and the ranking result, it also provides a guarantee in
the original sense proposed by \emph{Saaty} \cite{Saaty2005taha}.
Therefore, wherever the inconsistency index $\mathcal{S}(M)$ has
so far been used, $\mathcal{D}(M,\mu_{\maxid})$ might be used instead.
Provided of course, that $\mathcal{D}(M,\mu_{\maxid})$ is sufficiently
small. In return, in addition to the requirements of the level of
inconsistency, the users receive a guarantee of even discrepancy distribution.

\section{The ranking discrepancy and the conditions of order preservation}

Similarly as \emph{POP} and \emph{POIP} (Sec. \ref{sub:COP}), the
global ranking discrepancy is derived from eigenvalue based approach
but it does not depend on it. Thus, the definition (Def. \ref{def:local_inconsistency_def})
remains valid for any priority deriving method and any $\mu$. Moreover,
the value $\mathcal{D}(M,\mu)$ remains in the immediate connection
with \emph{POP} and \emph{POIP}. This relationship  could be expressed
in the form of the following two assertions. 
\begin{theorem}
\label{COP-theor}For every pairwise comparisons matrix $M$ expressing
the quantitative relationships $R$ between concepts $c_{1},\ldots,c_{n}\in C$,
and the ranking $\mu$, the order preference condition is preserved
i.e. 
\begin{equation}
m_{ij}>1\,\,\,\,\,\mbox{implies}\,\,\,\,\,\mu(c_{i})>\mu(c_{j})\label{eq:cop-proof_2}
\end{equation}

if wherever $\mathcal{D}(M,\mu)<\delta$ then also $m_{ij}\geq\delta+1$.\end{theorem}
\begin{proof}
Since $\mathcal{D}(M,\mu)<\delta$, then according to the definition
\ref{def:local_inconsistency_def}, every $\mathcal{E}(\widehat{i},\widehat{j})<\delta$
for $\widehat{i},\widehat{j}=1,\ldots,n$. In particular $\mathcal{E}(j,i)<\delta$,
hence also $\epsilon(j,i)-1<\delta$. Therefore, due to the definition
of $\epsilon$ (\ref{eq:e_def}) it is true that 
\begin{equation}
\frac{1}{m_{ji}}\cdot\frac{\mu(c_{j})}{\mu(c_{i})}<\delta+1\label{eq:proof_2-1}
\end{equation}

hence 
\begin{equation}
m_{ji}\frac{\mu(c_{i})}{\mu(c_{j})}>\frac{1}{\delta+1}\label{eq:proof_2-2}
\end{equation}

and due to the reciprocity 
\begin{equation}
\frac{\mu(c_{i})}{\mu(c_{j})}>\frac{m_{ij}}{\delta+1}\label{eq:proof_2-3}
\end{equation}

Therefore the ratio $\nicefrac{\mu(c_{i})}{\mu(c_{j})}$ is strictly
greater than one if only $\nicefrac{m_{ij}}{\delta+1}\geq1$. In other
words the only requirement in addition to $\mathcal{D}(M,\mu)<\delta$
needed to meet the \emph{POP} is $m_{ij}\geq\delta+1$. 
\end{proof}
The above theorem easily translates into an algorithm that allows
us to decide whether the pairwise comparison matrix $M$ and the ranking
$\mu$ are \emph{POP}-safe, i.e. whether the \emph{POP} condition
will never be violated for this pair. Let us note that if we adopt
a weak inequality as the upper bound of the ranking discrepancy index
i.e. $\mathcal{D}(M,\mu)\leq\delta$, then to meet the \emph{POP}
the strong inequality $m_{ij}>\delta+1$ is needed. Thus, assuming
that $\delta=\mathcal{D}(M,\mu)$ is known, all the ratios greater
than one i.e. $m_{ij}>1$ need to be examined to determine whether
they are also greater than $\delta+1$. If so, $M$ is \emph{POP}-safe,
which means that \emph{POP} is not violated. 

The relationship between \emph{POIP} and $\mathcal{D}(M,\mu)$ also
can be expressed in the form of assertion. 
\begin{theorem}
\label{COIP-theor}For every pairwise comparisons matrix $M$ expressing
the quantitative relationships $R$ between concepts $c_{1},\ldots,c_{n}\in C$,
and the ranking $\mu$, the order of intensity of preference condition
is preserved i.e. 
\begin{equation}
m_{ij}>m_{kl}>1\,\,\,\,\,\mbox{implies}\,\,\,\,\,\frac{\mu(c_{i})}{\mu(c_{j})}>\frac{\mu(c_{k})}{\mu(c_{l})}\label{eq:COIP_proof}
\end{equation}
if wherever $\mathcal{D}(M,\mu)<\delta$ then also $\nicefrac{m_{ij}}{m_{kl}}\geq\left(\delta+1\right)^{2}$\end{theorem}
\begin{proof}
Since $\mathcal{D}(M,\mu)<\delta$, then according to the definition
\ref{def:local_inconsistency_def}, every $\mathcal{E}(p,q)<\delta$
for $p,q=1,\ldots,n$. In particular $\mathcal{E}(j,i)<\delta$ and
$\mathcal{E}(k,l)<\delta$, hence also $\epsilon(j,i)-1<\delta$ and
$\epsilon(k,l)-1<\delta$. Thus, following the same reasoning as in
Theorem \ref{COP-theor} (\ref{eq:proof_2-1}, \ref{eq:proof_2-2}
and \ref{eq:proof_2-3}) we obtain that 

\begin{equation}
\frac{\mu(c_{i})}{\mu(c_{j})}>\frac{m_{ij}}{\delta+1}\,\,\,\,\,\mbox{and}\,\,\,\,\,\frac{\mu(c_{l})}{\mu(c_{k})}>\frac{m_{lk}}{\delta+1}\label{eq:Q1}
\end{equation}

hence due to the reciprocity, 
\begin{equation}
\frac{\mu(c_{i})}{\mu(c_{j})}>\frac{m_{ij}}{\delta+1}\,\,\,\,\,\mbox{and}\,\,\,\,\,\frac{\mu(c_{k})}{\mu(c_{l})}<m_{kl}\left(\delta+1\right)\label{eq:Q2}
\end{equation}

Therefore, dividing the left inequality by the right inequality leads
to the formula 

\begin{equation}
\frac{\frac{\mu(c_{i})}{\mu(c_{j})}}{\frac{\mu(c_{k})}{\mu(c_{l})}}>\frac{\frac{m_{ij}}{\delta+1}}{m_{kl}\left(\delta+1\right)}\label{eq:Q3}
\end{equation}

Therefore, the ratio $\nicefrac{\mu(c_{i})/\mu(c_{j})}{\mu(c_{k})/\mu(c_{l})}$
is greater than $1$ if $\nicefrac{m_{ij}/\left(\delta+1\right)}{m_{kl}\left(\delta+1\right)}$
is not smaller than $1$. In other words the truth of the following
inequality: 

\begin{equation}
\frac{m_{ij}}{m_{kl}}\geq\left(\delta+1\right)^{2}\label{eq:Q4}
\end{equation}

implies that
\begin{equation}
\frac{\mu(c_{i})}{\mu(c_{j})}>\frac{\mu(c_{k})}{\mu(c_{l})}\label{eq:Q5}
\end{equation}

which is the desired assertion. 
\end{proof}
Similar as before, to hold the above theorem it is enough for the
weak inequality $\mathcal{D}(M,\mu)\leq\delta$ and the strong inequality
$\nicefrac{m_{ij}}{m_{kl}}>\left(\delta+1\right)^{2}$ to hold. Thus,
for the practical verification of whether the \emph{POIP} is violated,
the condition $\nicefrac{m_{ij}}{m_{kl}}>\left(\delta+1\right)^{2}$
needs to be examined for every pair $m_{ij},m_{kl}$ that meets the
requirements of the theorem.

\section{Numerical example}

Let us consider a case of numerical judgment described in \cite{BanaeCosta2008acao}.
There are four concepts $c_{1},\ldots,c_{4}$ for which the relative
importance determined by a person $J$ is given as the matrix $M$. 

\begin{equation}
M=\left(\begin{array}{cccc}
1 & 2.5 & 4 & 9.5\\
0.4 & 1 & 3 & 6.5\\
\frac{1}{4} & \frac{1}{3} & 1 & 5\\
\frac{1}{9.5} & \frac{1}{6.5} & \frac{1}{5} & 1
\end{array}\right)\label{eq:example_M}
\end{equation}
The rescaled eigenvector corresponding to the maximal eigenvalue of
$M$ is given as: 

\begin{equation}
\mu_{\textit{max}}=\left[0.533,\,0.287,\,0.139,\,0.041\right]^{T}\label{eq:example_mu}
\end{equation}
As already pointed in \cite{BanaeCosta2008acao} \emph{POIP} is not
satisfied. In particular $m_{3,4}>m_{1,3}$ but $\mu_{\textit{max}}(c_{3})/\mu_{\textit{max}}(c_{4})<$
\\ $\mu_{\textit{max}}(c_{1})/\mu_{\textit{max}}(c_{3})$. The local
discrepancy matrix $\mathcal{E}=[\mathcal{E}(i,j)]$ allows for identifying
the most inconsistent entry in $M$. It is $m_{3,4}$, for which $\mathcal{E}(3,4)=0.475$.

\begin{equation}
\mathcal{E}=\left(\begin{array}{cccc}
0 & 0.348 & 0.044 & 0.367\\
0.348 & 0 & 0.452 & 0.077\\
0.044 & 0.452 & 0 & 0.475\\
0.367 & 0.077 & 0.475 & 0
\end{array}\right)\label{eq:example_E_matrix}
\end{equation}
After re-evaluation by experts the value $m_{3,4}$ is set to $3$.
Re-creating the local discrepancy matrix for $M$ where $m_{3,4}=3$
and $m_{4,3}=\nicefrac{1}{3}$ indicates that $m_{1,2}$ also needs
expert attention. Re-evaluated $m_{1,2}$ is set to $1.5$ and due
to the reciprocity requirement $m_{2,1}$ is set to $\nicefrac{2}{3}$.
After adjusting four entries the matrix (\ref{eq:example_M}) takes
the form:

\begin{equation}
M^{'}=\left(\begin{array}{cccc}
1 & 1.5 & 4 & 9.5\\
\frac{1}{1.5} & 1 & 3 & 6.5\\
\frac{1}{4} & \frac{1}{3} & 1 & 3\\
\frac{1}{9.5} & \frac{1}{6.5} & \frac{1}{3} & 1
\end{array}\right)\label{eq:example_M_prim}
\end{equation}
The rescaled principal eigenvector of $M'$ is: 

\begin{equation}
\mu_{\textit{max}}^{'}=\left[0.487,\,0.338,0.126,\ 0.048\right]^{T}\label{eq:example_mu_prim}
\end{equation}
The local discrepancy matrix $\mathcal{E}'=[\mathcal{E}'(i,j)]$ calculated
for $M'$ and $\mu_{\textit{max}}^{'}$ shows that the global ranking
discrepancy is $0.149$. 

\begin{equation}
\mathcal{E}^{'}=\left(\begin{array}{cccc}
0 & 0.038 & 0.033 & 0.064\\
0.038 & 0 & 0.119 & 0.077\\
0.033 & 0.119 & 0 & 0.149\\
0.064 & 0.077 & 0.149 & 0
\end{array}\right)\label{eq:example_E_matrix_prim}
\end{equation}
According to the Theorem \ref{COP-theor} to meet the \emph{POP} condition
it is enough if 
\begin{equation}
m_{ij}^{'}>1\Rightarrow m_{ij}^{'}>1.149\label{eq:example_COP1}
\end{equation}
 for every $i,j=1,\ldots,4$ and $M'=[m_{ij}^{'}]$. Similarly, (Theorem
\ref{COIP-theor}) the \emph{POIP} condition is satisfied if 
\begin{equation}
m_{ij}^{'}>m_{kl}^{'}>1\Rightarrow m_{ij}^{'}/m_{kl}^{'}>(1+0.149)^{2}\approx1.32\label{eq:example_COP2}
\end{equation}
for every $i,j,k,l=1,\ldots,4$. It is easy to see that both (\ref{eq:example_COP1})
and (\ref{eq:example_COP2}) hold. Therefore, after the discrepancy
reduction%
\footnote{and inconsistency reduction. Note that $\mathscr{S}(M)=0.04$ whilst
$\mathscr{S}(M')=0.003$.%
}, there is a guarantee that the resulting pairwise comparisons matrix
$M'$ together with the ranking $\mu_{\textit{max}}^{'}$ satisfies
\emph{COP}.

\section{Discussion and summary}

In their work \emph{Bana e} \emph{Costa }and \emph{Vansnick} \cite{BanaeCosta2008acao}
formulated two conditions whose fulfillment makes the ranking result
indisputable. Therefore, in practice, meeting these two conditions
may translate into a significant reduction in the number of appeals
against the results of the ranking procedure. Hence, in addition to
intangible benefits such as providing the ranking participants a sense
of justice, meeting the \emph{POP} and \emph{POIP }conditions may
contribute to the reduction of costs associated with the carrying
out the evaluation procedure. The notion of global ranking discrepancy
$\mathcal{D}(M,\mu)$ helps to fulfill the \emph{Bana e} \emph{Costa
}and \emph{Vansninck} postulate\emph{.} The value $\mathcal{D}(M,\mu)$
directly translates to the requirements for the matrix $M$, so that
the smaller $\mathcal{D}(M,\mu)$ the greater the chance that the
\emph{POP} and \emph{POIP} conditions are met. 

Although the global ranking discrepancy (Sec. \ref{sec:The-ranking-discrepancy})
has been defined in the context of eigenvalue priority deriving method,
it is not tied to it. In fact it could be useful for any pair of the
\emph{PC} matrix $M$ and the ranking $\mu$. The only, but crucial,
assumption is that $\mu$ attempts to reflect the experts' judgments
given as $M$. The conditions provided by Theorems (\ref{COP-theor})
and (\ref{COIP-theor}) are sufficient, but they are not necessary.
Thus, there may exist better estimates allowing to determine whether
the \emph{COP} are satisfied. The existence of such estimates remains
as an open question. 

This study addresses an important problem of discrepancies between
expert judgments and ranking results that may appear in the pairwise
comparison method. The notion of the global ranking discrepancy has
been defined. Its relationship with the eigenvalue based inconsistency
index and \emph{POP} and \emph{POIP} \cite{BanaeCosta2008acao} postulates
have been shown. 
\begin{acknowledgements}
I would like to thank Dr Jacek Szybowski and Prof. Antoni Lig\k{e}za
for reading the first version of this work. Special thanks are due
to Dan Swain for his editorial help. The research is supported by
AGH University of Science and Technology, contract no.: 10.10.120.105.
\end{acknowledgements}
\bibliographystyle{plain}
\bibliography{papers_biblio_reviewed}

\end{document}